\RequirePackage{fix-cm}
\documentclass[smallextended]{svjour3}       % onecolumn (second format)
\smartqed  % flush right qed marks, e.g. at end of proof
\usepackage{graphicx}
\usepackage{amsfonts}
\usepackage{amsmath}
\usepackage{url}
\usepackage{a4wide}
\makeatletter
\let\proof\@undefined
\let\endproof\@undefined
\makeatother

\usepackage{amsthm}
\usepackage{epstopdf}
\usepackage{color}
\usepackage{subfigure}

\newcommand{\cX}{{\mathcal X}}

\newcommand{\cC}{{\mathcal C}}
\newcommand{\cT}{{\mathcal T}}

\newcommand{\cS}{{\mathcal S}}
\newcommand{\cW}{{\mathcal W}}

\newcommand{\stevenc}[1]{\textcolor{black}{#1}}

\newcommand{\cs}[1]{\textcolor{black}{#1}}

\usepackage[ruled,vlined,english]{algorithm2e}

\begin{document}

\title{Towards the fixed parameter tractability of constructing
minimal phylogenetic networks from arbitrary sets of nonbinary trees
}


\author{Steven Kelk         \and
        Celine Scornavacca %etc.
}

\institute{S. Kelk \at
              Department of Knowledge Engineering (DKE), Maastricht University, P.O. Box 616, 6200 MD Maastricht, The Netherlands.  \\
              Tel.: +31 (0)43 38 82019\\
              Fax: +31 (0)43 38 84910 \\
              \email{steven.kelk@maastrichtuniversity.nl}           %  \\
           \and
           C. Scornavacca \at
Institut des Sciences de l'Evolution (ISEM, UMR 5554 CNRS),
Universit\'e Montpellier II, Place E. Bataillon - CC 064 - 34095 Montpellier Cedex 5, France,
\\
              \email{celine.scornavacca@univ-montp2.fr} 
}

\titlerunning{Towards the fixed parameter tractability of nonbinary trees}

\date{Received: date / Accepted: date}

\maketitle

\begin{abstract}
It has remained an open question for some time whether, given a set of not necessarily binary (i.e. ``nonbinary'') trees $\cT$ on a set of taxa $\cX$, it is possible to determine in
time $f(r) \cdot poly(m)$ whether there exists a phylogenetic network that displays \stevenc{all} the trees in $\cT$, where $r$ refers to the reticulation number of the network and $m=|\cX|+|\cT|$. Here we show that this holds if one or both of the following conditions holds:  (1) $|\cT|$ is bounded by a function of $r$; (2) the maximum degree of the nodes in $\cT$ is bounded by a function of $r$. These
sufficient conditions absorb and significantly extend known special cases, namely when all the trees in $\cT$ are binary \cite{ierselLinz2012} or $\cT$ contains exactly two nonbinary trees \cite{linzsemple2009}. We believe this result is an important step towards settling the issue for an arbitrarily
large and complex set of nonbinary trees. For completeness we show that the problem is certainly solveable in time $O( m^{f(r)} )$.
\end{abstract}

\section{Introduction}

A rooted phylogenetic tree  on a set of taxa $\cX$ is a directed tree in which exactly one node has indegree zero (the \emph{root}), all edges are directed away from the root, the leaves
are bijectively labelled by $\cX$ and there are no nodes with indegree and outdegree  both equal to one. Rooted phylogenetic trees, henceforth \emph{trees}, are used to model the evolution
of the set $\cX$ starting from a (distant) common ancestor, the root  \cite{SempleSteel2003,MathEvPhyl,reconstructingevolution}. Rooted phylogenetic networks, henceforth \emph{networks}, are a generalisation of trees which allow a wider array of
evolutionary phenomena to be modelled, specifically those phenomena which involve the convergence, rather than divergence, of lineages. For detailed background information on networks we refer the reader to \cite{husonetalgalled2009,HusonRuppScornavacca10,surveycombinatorial2011,twotrees,Nakhleh2009ProbSolv,Semple2007}.

A network on $\cX$ is a directed acyclic graph with  a unique node of indegree zero (the root) from which all other nodes in the graph can be reached by a directed path, the leaves are bijectively labelled by $\cX$ and there
are no nodes of indegree and outdegree one (see Figure \ref{fig:network}). Of particular interest are nodes with indegree two or higher, called \emph{reticulations}. The reticulation number $r(N)$ of a network is the sum of the indegrees of the reticulation nodes, minus the total number of reticulation nodes. It is the reticulation nodes that allow us to  simultaneously ``embed'' multiple trees
(evolutionary hypotheses) into the network, a classical biological motivation being the embedding of multiple discordant gene trees into a single species network \cite{Nakhleh2009ProbSolv}. 
\begin{figure}[h]
\centering
\includegraphics[scale=.3]{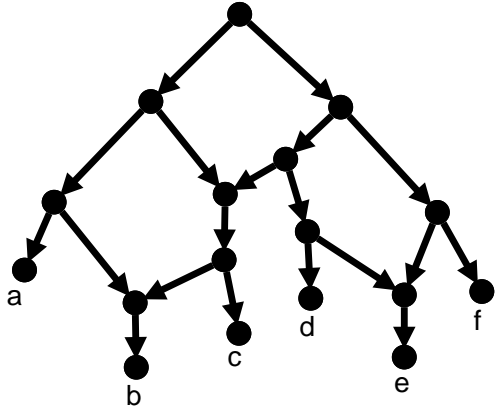}
\caption{A (binary) rooted phylogenetic network $N$ on $\cX =\{a,\ldots,f\}$. This network
has three reticulations.}
\label{fig:network}
\end{figure}
\stevenc{More formally, we say that a network $N$ on $\cX$ \emph{displays} a tree $T$ on $\cX$ if a subtree $T'$ of $N$ exists such that: (1) $T'$ is obtained by, for each reticulation in $N$, deleting all but one of its incoming edges, and (2) $T$ can be obtained from $T'$ by contracting some subset of the edges of
$T'$ (see Figure \ref{fig:display}). If every internal node of $T$ has outdegree two we say that $T$ is \emph{binary}. When we say that a tree is
\emph{nonbinary} we simply mean ``not necessarily binary''. (A binary tree is thus also a nonbinary tree). Note that in the case that $T$ is nonbinary the notion ``displays'' permits that the image of $T$ inside $N$ is more resolved than $T$ itself. This is motivated by the fact that biologists often use nodes with outdegree 3 or higher to indicate \emph{uncertainty}, rather than a hard topological constraint \cite{HusonRuppScornavacca10,davidbook}. On the other hand, if $T$ is binary, ``displays'' allows no such freedom, since $T$ is already completely resolved.}
\begin{figure}[t]
\centering
\includegraphics[scale=.3]{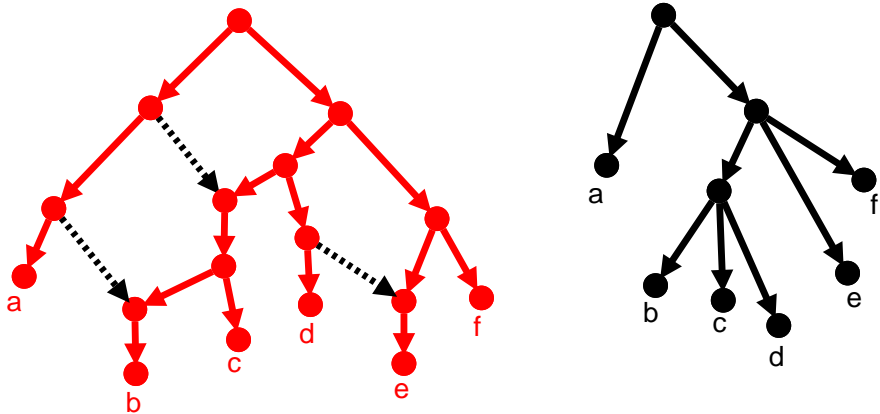}
\caption{Here we see that the network $N$ from Figure \ref{fig:network} displays the
tree $T$ shown here on the right. The dotted lines denote the reticulation edges that should
be deleted. Note that here the image of $T$ inside $N$ is more resolved than $T$ itself.}
\label{fig:display}
\end{figure}

 In recent years there has been much attention for the following optimization problem,
motivated by the desire in biology to postulate a ``most parsimonious'' network which can simultaneously explain a set of evolutionary hypotheses modelled as trees. \emph{Given a set $\cT$ of trees on the same set of taxa $\cX$, construct
a network $N$ on $\cX$ that displays each of the trees in $\cT$ with as small reticulation number as possible}. This question has stimulated a considerable amount of mathematical research, with
most attention thus far going to the case when $\cT$ consists of two binary trees. Even this restricted variant of the problem is NP-hard, APX-hard  \cs{\cite{bordewich}}  and possibly not even in APX, having similar
(in)approximability properties to the classical problem feedback vertex set on directed graphs \cs{\cite{approximationHN}}. Despite the worrying approximability news there has been considerable progress on the question
of fixed parameter tractability (FPT), where the parameter in question is the minimum reticulation number, denoted $r(\cT) = r$.  (For background information on FPT, see \cite{niedermeier2006,Flum2006,fptphylogeny}).  For two binary trees a suite of different (but related) FPT algorithms and tools \cite{hybridinterleave,hybridnet,fastcomputation},
have been developed, the theoretical state-of-the-art being $O( 3.18^{r}n )$ where $n = |\cX|$ \cite{whiddenFixed}, and a quadratic kernel \cite{sempbordfpt2007,bordewich2,quantifyingreticulation}. For more than two trees, or when $\cT$ contains nonbinary trees, there is much less known. In \cite{linzsemple2009} it is shown that, if $\cT$ contains two nonbinary trees, the problem is also FPT (also via a quadratic kernel), although there are considerably more technicalities involved than in the binary case. Very recently  \cite{ierselLinz2012} showed how to construct a quadratic kernel for an arbitrary number of binary trees and \cite{simplefpt} gave a simplified bounded-search FPT algorithm for two nonbinary trees. This begs the question: is the problem still FPT for an arbitrary number of
nonbinary trees? Although we do not yet have a full answer to this, we have identified quite broad conditions on $\cT$ under which an FPT algorithm is possible, which absorbs and extends the
conditions on $\cT$ posed by \cite{linzsemple2009} and \cite{ierselLinz2012}. Specifically, a set of \stevenc{nonbinary} trees $\cT$ on the same set of taxa $\cX$ is \emph{well-bounded} if at least one of the following two conditions hold, where $f(r)$ is used to denote a function that depends only on $r$ and not on the size of the input: (1) there are at most $f(r)$ trees in $\cT$; (2) the maximum degree ranging over all nodes in all trees in $\cT$ is at most $f(r)$. Clearly, the case solved by \cite{linzsemple2009} implies well-boundedness because $|\cT| \leq 2$, and the case solved by \cite{ierselLinz2012}
implies well-boundedness both because the maximum degree of any node is 3 \emph{and} because without loss of generality it can be assumed that $|\cT| \leq 2^{r}$ when all the trees in $\cT$
are binary. That is, sets of binary trees obey both the possibilities for well-boundedness, despite only one being necessary, and this gives clues as to the comparative tractability of the binary case.
Note that, when trees are permitted to be nonbinary, there is no obvious $f(r)$ upper-bound on the size of $\cT$.

In this article we give a constructive \stevenc{bounded-search}
algorithm which is FPT whenever $\cT$ is well-bounded. We prove this by extending a related FPT result from \cite{softwiredClusterFPT}. In that article the input is a set of \emph{clusters}, where a cluster is simply a \stevenc{subset} of $\cX$. We say that a tree $T$ on $\cX$ represents a cluster $C \subseteq \cX$ whenever $T$ contain an edge $(u,v)$ such that the set of taxa reachable in $T$ by directed paths starting from $v$, is equal to $C$.
A network $N$ on $\cX$ \emph{represents} a cluster $C$ whenever there exists \emph{some} tree $T$ on $\cX$ with the following properties: (1) $N$ displays $T$ and (2) $T$ represents $C$. Given
a set $\cT$ of trees on $\cX$, it is natural to define the set $Cl(\cT)$ as the set of all clusters represented by some tree in $\cT$. In \cite{softwiredClusterFPT} it is shown that computing
a network $N$ with minimum reticulation number that represents a set of clusters $\cC$, is FPT, again using reticulation number as the parameter. The question immediately
arises: what if we apply the result from \cite{softwiredClusterFPT} taking $\cC = Cl(\cT)$? Unfortunately, there are cases when the optimum under the cluster model can be strictly lower
than the optimum under the tree model \cite{twotrees}, which stems from the fact that a network $N$ might represent all the clusters in $Cl(\cT)$ but not display all the trees in $\cT$.  
However, it is tempting to ask whether the bounded-search algorithm given in \cite{softwiredClusterFPT} can be adapted by replacing \stevenc{intermediate} tests of the form \emph{``does $N$ represent $Cl(\cT)$?''} with
\emph{``does $N$ display all the trees in $\cT$?''}. Here we show that in many cases the answer to this question is \emph{yes}. However, there seem to be some pathological cases where the
stronger topological demands of the tree model cannot easily be captured by the cluster model. These pathological cases are excluded by our definition of well-boundedness, and the
major open question emerging from this paper is whether these pathological cases are genuinely more difficult than the well-bounded case. We suspect that they can be overcome, because given
a parameter value $r$ and an arbitrarily large set of trees $\cT$ with arbitrarily large maximum degree, the parameter $r$ does impose quite severe restraints on the topology of the trees in $\cT$
and the location of taxa in an optimum network $N$. We discuss these possibilities at the end of the article. Finally, for completeness
we also give a simple non-FPT algorithm which\cs{,} for fixed $r$\cs{,} determines in polynomial-time whether $r(\cT) \leq r$ (and if so constructs an appropriate network), establishing that even inputs that are not well-bounded can be solved in polynomial time.

Given the very close relationship between this article and \cite{softwiredClusterFPT} we have decided not to repeat all 
%definitions and 
algorithms from \stevenc{that article. For this reason it is necessary to read \cite{softwiredClusterFPT} before, or in parallel, with this article.}
The algorithmic changes are small; we only have to adapt two steps in a much larger algorithmic procedure. Furthermore, it is relatively straightforward to argue that this adapted algorithm
is correct and definitely constructs a network $N$ that displays all the trees in $\cT$ such that $r(N) = r(\cT)$.  However, demonstrating that the running time is FPT requires much more work, and is the focus of this article. We are forced to deal with the aforementioned pathological situations
by exhaustive guessing, and well-boundedness guarantees that the branching in the search tree caused by this pessimistic guessing does not spiral out of control.

%We should emphasize that this result really is significantly more general than
%\cite{ierselLinz2012} and \cite{linzsemple2009}. For example, our result proves
%FPT for any constant number of nonbinary trees, or any arbitrarily large set of nonbinary
%trees where the maximum degree is bounded. It also proves FPT for any arbitrary set
%of binary trees. We should emphasize that binary
%is comparatively ``easy'' and that a peculiarity of nonbinary is that there can be
%a huge number of trees in the input. I was thinking that a crude upper bound on $\cT$,
%expressed purely as a function of $r$ and $|\cX|$, is $2^{r}2^{2(|\cX|-1)}$ i.e. take
%any binary refinement and contract some arbitrary subset of its edges. It might be
%necessary to include our old examples which show that we really need to take
%$|\cT|$ into account in the size of the input.

\section{Preliminaries}

%\marginpar{This section will need updating and making consistent,  the concepts ``display'' \cs{side minimal and valid completion} in particular.}

\cs{Some of the basic definitions, e.g. phylogenetic tree and phylogenetic network, have already been given in the introduction. In this section we will introduce several other definitions that will be used in the rest of the article. }

\cs{A network is said to be \emph{binary} if every reticulation node has indegree 2 and outdegree 1 and every other interior node has outdegree 2. A \emph{(binary) refinement} of a tree $T$ is any (binary) tree $T'$ such that $Cl(T) \subseteq Cl(T')$.} \stevenc{(Note that every tree is a refinement
of itself).}

%%%%%%%
Given  a set of taxa $\cX$, we say that two clusters $C_1, C_2 \subset \cX$ are \emph{compatible} if either $C_1 \subseteq C_2$ or $C_1 \supseteq C_2$ or $C_1  \cap C_2=\emptyset$, and incompatible otherwise.  We say that a set of taxa $S \subseteq \cX$ is compatible with a cluster set $\cC$ if every cluster $C \in \cC$ is compatible with $S$, and incompatible otherwise.
We say that a set $S\subseteq \cX$ is an \emph{ST-set} with respect to a set of clusters $\cC$, if $S$ is compatible with $\cC$ and the clusters of $\cC|S$ are pairwise compatible, \stevenc{where $\cC|S$ is defined as the set of clusters
$\{ C \cap S | C \in \cC \}$.} An ST-set $S$ is said to be  \emph{maximal} if there is no ST-set $S' $ with $S \subset S'$. Given two taxa $x,y \in \cX$, we write $x  \rightarrow_\cC y$, if and only if every non-singleton cluster in $\cC$ containing $x$, also contains $y$.

%%%%%%%
\cs{An $r$-reticulation generator is defined as a directed acyclic multigraph, which has a single node of indegree 0 and outdegree 1, precisely $r$ reticulation nodes, and apart from that only nodes of indegree 1 and outdegree 2 \cite{softwiredClusterFPT}. The \emph{sides} of a $r$-reticulation generator are defined as the union of its edges (the \emph{edge sides}) and its reticulation nodes \stevenc{of outdegree 0} (the \emph{node sides}). 
Adding a set of taxa $X$ to an edge side $(u,v)$ of an $r$-reticulation generator consists of subdividing $(u,v)$ into a path of $|X|$ internal nodes and, for each such internal node $w_i$, adding a new leaf $w_i'$ along with an edge $(w_i,w_i')$, and labeling $w_i'$ with some taxon from $X$ in such a way that $X$ bijectively labels the new leaves. On the other hand, adding a taxon $l$ to a node side $v$ consists of adding a new leaf $y$ along with an edge $(v, y)$ and labeling $y$ with $l$.
Given a set of taxa $\cX$, the set $\mathcal{\hat{N}}^{r}_{\cX}$ is defined as the set of all networks that can be constructed by choosing some $r$-reticulation generator $G$ and then adding taxa to the sides of $G$ as described above  such that each taxon of $\cX$ appears exactly once in the resulting network $N$ and  no multi-edge is present  in $N$. (Note that a \stevenc{``fake root''}, i.e. a  single node with indegree 0 and outdegree 1, \stevenc{will still be} present in $N$,
\stevenc{but this can simply be deleted along with its incident edge.}) The resulting network $N$ is said to be a \emph{completion} of $G$ on $\cX$ while $G$ is said to be the generator \emph{underlying} $N$. 
}

Given a network $N$, we say that a \emph{switching} $T_N$ of $N$ is obtained by, for each reticulation node,
deleting all but one of its incoming edges. (The red subtree shown in Figure \ref{fig:display} is
a switching). The definition of \emph{display} given in the
introduction is thus equivalent to: a network $N$ on $\cX$ \emph{displays} a tree $T$ on $\cX$ if $N$ has some switching $T_N$ such that $T$ can be obtained from $T_N$ by contracting some subset of the edges of $T_N$. (In which case we say that $T_N$ \emph{corresponds} to $T$). 
For a set of trees $\cT$, the \emph{hybridization number} of $\cT$  is defined as  the minimum reticulation number of $N$, ranging over all networks $N$ that display all the trees in $\cT$.
For consistency with \cite{softwiredClusterFPT} we henceforth say \emph{reticulation number} of $\cT$, denoted $r(\cT)$, instead of hybridization number.

We note that the value $r(\cT)$ does not change if, in the definition above, we restrict
our attention to binary networks $N$. This follows because there is a simple construction
described in \cite[Lemma 2]{twotrees} which transforms a network $N$ that displays all
the trees in $\cT$ into a binary network $N'$ that displays all the trees in $\cT$, such
that $r(N') = r(N)$. For this reason we can
restrict our attention to binary networks. Note that, for a binary network $N$ and a tree
$T$, the statement ``$N$ displays $T$'' is equivalent to the statement ``$N$ displays
a binary refinement of $T$''.

\section{Minimizing \stevenc{the} reticulation number of well-bounded sets of trees is fixed parameter tractable}

\stevenc{Henceforth, unless otherwise stated, the parameter $r$ is obtained from the question \emph{``Is $r(\cT) \leq $r?''}. Clearly an algorithm to solve this problem can be used to solve the corresponding optimization problem i.e. by
solving the decision problem for increasing values of $r$ until $r(\cT)$ is reached.}

\stevenc{We begin by repeating} the definition of \emph{well-boundedness} already mentioned in the introduction.
\begin{definition}
A set of \stevenc{nonbinary} trees $\cT$ on the same set of taxa $\cX$ is \emph{well-bounded} if at least one of the
following two conditions hold:  (1) there are at most $f(r)$ trees in $\cT$; (2) the maximum degree ranging over all nodes in all trees in $\cT$ is at most $f(r)$.
\end{definition}

Note that if all the trees in $\cT$ are binary then $\cT$ is trivially well-bounded thanks to
condition (2). However, without loss of generality we can actually also assume that when all trees in $\cT$ are binary condition
(1) holds. This is because a (binary) network with $r$ reticulations
can display up to $2^{r}$ binary trees. It is therefore pointless to try constructing
a network with $r$ reticulations if $|\cT| > 2^r$. This places a natural $f(r)$ upper bound
on $|\cT|$ in the binary case.\\

The next result \stevenc{shows that claimed solutions
can be efficiently verified. Our main algorithm will make heavy use of this fact to prune the search space.}

 %Let us start by supposing that our collection of trees $\cT$ is composed of \emph{binary} trees. Then, by Lemma 3 of \cite{twotrees}, we know that we can focus on binary trees. Indeed, this lemma asserts, among other things, that if a network $N$ displaying $\cT$ exists, then there exists also a binary network $N'$  displaying $\cT$ such that $r(N')\leq r(N)$. The following result is straightforward:

%\marginpar{\tiny Say why we cannot have $|\cX|=n$ and show counter example. Say that $\cT$ will not be too big normally....}
\begin{proposition}
\label{prop:checkReprTree}
Given a binary network $N$ with reticulation number $r$ and a set of trees $\cT$ on $\cX$, checking whether  $N$ displays all the trees in $\cT$  can be done in  time $f(r) \cdot poly(m)$, where $m=|\cX|+|\cT|$.
\end{proposition}
\begin{proof}
Since $N$ is binary, the set $\cT(N)$ of \cs{binary} trees displayed by $N$ has cardinality at most $2^{r}$. Each tree in $\cT$ and $\cT(N)$ contains at most $2(n-1)$ clusters, \stevenc{where $n=|\cX|$.} To determine if a tree $T$ is
a refinement of a tree $T'$ we need only check that $Cl(T') \subseteq Cl(T)$ which can be done in $poly(n)$ time. In total therefore at most $|\cT| \cdot |\cT(N)| \cdot poly(n)$ time is required, which is $f(r) \cdot poly(m)$.
\end{proof}

Note that, since $|\cT|$ can be exponential in $|\cX|=n$ (see the Appendix), the result does not hold if $m=|\cX|$. This motivates our choice of $m = |\cX| + |\cT|$ as our measure
of input size.\\
\\
Before proceeding, note that a binary network $N$ on $\cX$ with $r$ reticulations can represent at most $2^{r+1}(n-1)$ clusters. \stevenc{This follows from the proof of Proposition \ref{prop:checkReprTree}}. If $N$ displays all the trees in $\cT$ then
it also represents all the clusters in $Cl(\cT)$. Hence, if $|Cl(\cT)| > 2^{r+1}(n-1)$ we can immediately conclude that $N$ does not display all the trees in $\cT$, we call this the \emph{cluster bound}. We may therefore henceforth
assume that $|Cl(\cT)| \leq 2^{r+1}(n-1)$, which is $f(r).poly(n)$. \stevenc{In any case,
$|Cl(\cT)|$ is at most $poly(m)$, because $|Cl(\cT)| \leq |\cT| \cdot 2(n-1)$.}\\
%$Cl(\cT)$ itself can easily be constructed in $poly(m)$ time.\\
%, where $m = |\cX| + |\cT|$.\\
\\
Let $\cT$ be a set of nonbinary trees on $\cX$ and let $\cS$ be the set of maximal ST-sets of $Cl(\cT)$.
Given an ST-set $S\in \cS$,  \emph{reducing} $S$ in $\cT$ denotes the operation of, for each tree $T$ in $\cT$,  adding a single leaf with a new label $s$ (the same in all trees of $\cT$) as a child node of  lca$_T(S)$, deleting all labels of $S$ in $T$ and finally applying in an arbitrary order
the following steps until no more can be applied:  
(a) deleting all nodes with outdegree-0 that are not labelled
by a taxon; (b) suppressing any node with indegree-1 and outdegree-1. 
We say that $\cT$ is \emph{ST-collapsed} if all maximal ST-sets of $Cl(\cT)$ have size 1\footnote{This is consistent with the definition given in \cite{softwiredClusterFPT}: a
set of clusters $\cC$ on $\cX$ is ST-collapsed if every maximal ST-set of $\cC$ has size 1.}.

We denote with the term \emph{ST-collapsing} the operation of reducing in a set of trees $\cT$ all maximal ST-sets of $Cl(\cT)$. This is always possible because maximal ST-sets are disjoint \stevenc{(and unique)} \cite[Corollary 4]{elusiveness}. \stevenc{The maximal
ST-sets can be computed in time $poly(|\cX|, |Cl(\cT)|)$ \cite{elusiveness}, from which it can be seen that ST-collapsing
all the trees in $\cT$ can be performed in time $poly(m)$}. Note that the set of trees $\cT'$ obtained by ST-collapsing $\cT$, and the associated cluster set $Cl(\cT')$ are ST-collapsed\footnote{\stevenc{Essentially, ST-collapsing $Cl(\cT)$ is the operation of
identifying and reducing all the ``maximal common pendant subtrees'' of $\cT$, see \cite{simplefpt,elusiveness} for a more detailed discussion of this.}}.

\begin{lemma}
\label{lem:stcollapse1}
Let $\cT$ be a set of trees  on $\cX$, and let $\cT'$ be the set of trees obtained by ST-collapsing $\cT$. Then any network $N'$ that displays all the trees in $\cT'$ can be transformed into a network
$N$ that displays all the trees in $\cT$ such that $r(N)=r(N')$, in \stevenc{$poly(m)$} time.
\end{lemma}
\begin{proof}
Let $\cS = \{S_1, \ldots, S_l\}$ be the set of maximal ST-sets of $Cl(\cT)$. For each $S_j \in \cS$ we replace the dummy taxon $s_j$ in $N'$ with a binary tree on taxon set $S_j$ that represents the set of clusters $Cl(\cT)|S_j$. The obtained network $N$ obviously \stevenc{displays} all the trees in $\cT$. By Corollary 4 of \cite{elusiveness}, maximal ST-sets are disjoint and $|\cS|$ is at most $n$.  Since \stevenc{$Cl(\cT)$ contains at most $poly(m)$} clusters, the entire transformation can be performed in $poly(m)$ time.
\end{proof}

\begin{lemma}
\label{lem:moveST}
%\marginpar{I didn't check this in detail but I believe it. We should consider moving this to the \stevenc{A}ppendix}
Given a set of trees $\cT$ on $\cX$, let $N$ be a network displaying all the trees in $\cT$. Let $S$ be a \stevenc{non-singleton} ST-set
with respect to $Cl(\cT)$. Then there exists a network $N'$ displaying all the trees in $\cT$ such that $r(N') \leq r(N)$, %$\ell(N') \leq \ell(N)$,
$S$ is under a cut-edge in $N'$ and for each ST-set $S'$ such that $S' \cap S = \emptyset$ and $S'$ is
under a cut-edge in $N$, $S'$ is also under a cut-edge in $N'$.
\end{lemma}
The proof of the lemma can be found in the \stevenc{A}ppendix.
The following corollary \stevenc{stems} from the fact that maximal ST-sets are disjoint:
\begin{corollary}
\label{lem:stcollapse2}
Let $N$ be a network displaying all the trees in $\cT$. There exists a network $N'$
 displaying all the trees in $\cT$ such that $r(N') \leq r(N)$ %, $\ell(N') \leq \ell(N)$ 
and all maximal ST-sets (with respect to $Cl(\cT)$)
are below cut-edges.
\end{corollary}

\begin{lemma}\label{lem:hybridNumInvariant}
Let $\cT$ be a set of trees on $\cX$, and let $\cT'$ be the set of trees obtained by ST-collapsing $\cT$. Then $\stevenc{r(\cT')} = r(\cT)$ \stevenc{and optimal solutions for $\cT'$ can be converted into optimal solutions for $\cT$ in time
$poly(m)$}.
\end{lemma}
\begin{proof}
This easily follows from Lemma \ref{lem:stcollapse1} and Corollary \ref{lem:stcollapse2}. Indeed, the former result ensures that $r(\cT) \leq r(\cT')$ while the latter that $r(\cT) \geq r(\cT')$. 
\end{proof}

\stevenc{Lemma \ref{lem:hybridNumInvariant} shows that, without loss of generality, we can restrict our attention
to ST-collapsed sets of trees.} The next lemma ensures that we can \stevenc{consequently} narrow our search to the networks in $\mathcal{\hat{N}}^{r(\mathcal{T})}_{\cX}$.% (for the definition of $\mathcal{\hat{N}}^{r(\mathcal{T})}$, see Definition 4 of \cite{softwiredClusterFPT}):
\begin{lemma}
\label{lem:genIsSufficient_bis}
Let $\cT$ be an ST-collapsed set of trees on $\cX$, such that $r(\cT) \geq 1$. Then there exists a network $N$ in $\mathcal{\hat{N}}^{r(\mathcal{T})}_{\cX}$ such that
$N$ displays all the trees in $\cT$. 
\end{lemma}
\begin{proof}
Let $N$ be any binary network with reticulation number $r(\cT)$ such that
$N$ displays all the trees in $\cT$. 
To prove the result, we need to prove that applying the reverse of the transformation
described in Definition 4 of  \cite{softwiredClusterFPT} to $N$ will give some $r(\cT)$-reticulation
generator $G$.
The proof is the same of that of  Lemma 4 (extended) of \cite{softwiredClusterFPT}. \stevenc{(This follows
because that lemma only considers the topology of $N$ and a sufficient pre-condition for the lemma to hold
is that $N$ represents an ST-collapsed set of clusters, which is certainly the case here because $N$ satisfies
the stronger requirement of displaying all the trees in $\cT$.)}
%the proof 
%Indeed,  since $Cl(\cT)$ is ST-collapsed and any network displaying $\cT$ represents $Cl(\cT)$,  we can apply Lemma 8 of %\cite{softwiredClusterFPT} in the proof as done in \cite{softwiredClusterFPT}. 
\end{proof}
\cs{
We call a side of a generator $G$ \emph{long}, \emph{short} or \emph{empty} if it is allowed to receive respectively $\geq$ 2  taxa, 1 taxon or 0 taxa. We call a \emph{set of side guesses} for a generator $G$, denoted by $S_G$, an assignment of  
type (empty, short, or long) to each side of $G$. A completion $N$ of $G$ on $\cX$ \emph{respects} a set of side guesses $S_G$ if 
{long}, {short} or {empty} sides in $S_G$ respectively received $\geq$ 2  taxa, 1 taxon or 0 taxa in $N$. A pair $(G,S_G)$ is said to be \emph{side-minimal} w.r.t. $\cT$ and $r$ if there exists a network $N$ such that (1) $r(N)=r$  (2) $N$ displays all the trees in $\cT$ (3) $N$  is a completion of $G$ on $\cX$ respecting $S_G$  and (4) $N$ has, amongst all binary networks displaying all the trees in $\cT$,  a minimum number of long sides, and (to further break ties) amongst those networks it has a minimum number of short sides. 
We define an \emph{incomplete network} as a generator $G$, a set of side guesses $S_G$, a set of \emph{finished} sides
(i.e. those sides to which we are not allowed to add taxa anymore), a set of
\emph{future} sides (i.e. short and long sides on which  no taxa has been placed yet) and at most one
\emph{active} side (i.e. a long side  not yet declared as finished to which we have already allocated at least one taxon). A  \emph{valid completion} of an incomplete network is an
assignment of the unallocated taxa to the future sides and (possibly) to the active side, that respects $S_G$ and such that the resulting network displays all the trees in $\cT$. 
}
%\pagebreak
%
\stevenc{We are now ready to give the main lemma of this article. It shows that,
once we have started on a new active side $s$, we know how to correctly continue adding
taxa to it, and when to declare that it is finished.}

\begin{lemma}
\label{lemma:subroutineIsCorrect}
Let $\cT$ be a \cs{well-bounded} ST-collapsed set of trees on $\cX$ and let $r$ be the first integer such that a network \cs{with reticulation number $r$ displaying all the trees in $\cT$ exists}. Let $N$ be an incomplete network such that 
its underlying $r$-reticulation generator $G$ and set of side guesses $S_G$ are such that $(G,S_G)$ is side-minimal w.r.t. \cs{$\cT$} and $r$, and let $s$ be an active side of $N$. 
Then, if a valid completion for $N$ exists, Algorithm \ref{algo:addOnSide}  computes a set of (incomplete) networks $\mathcal{N}$ such that  
 this set contains at least one  network  for which a valid completion exists \cs{in time $f(r) \cdot poly(m)$,  where $m=|\cX|+|\cT|$}. 
\end{lemma} 

%%%%%%%%%%%%

\begin{algorithm}[] 
\nlset{15} if $N(l, s)$ does not display each tree in $\cT$ restricted to $\cX(N)  \cup \{l\}$ then declare $s$ as finished and return $N$. If $x_i$ is currently
the only taxon on side $s$ then return $N(l,s)$. If all trees are \emph{safe} \stevenc{w.r.t. $s$} then return
$N(l,s)$. Otherwise, guess i.e. return both $N(l,s)$ and $N$ where $s$ is declared as finished.

\cs{\nlset{53}} %Leave lines 48-52 as they are and then replace line 53 with the following:
if $N^{*}(l, s)$ does not display each tree in $\cT$ restricted to $\cX(N)  \cup \{l\}$ then
guess either to end the side or to put $l$ somewhere in $U$. Otherwise, if all trees are \emph{safe} \stevenc{w.r.t. $s$} then
return $N(l,s)$. Otherwise (i.e. at least one tree is unsafe \stevenc{w.r.t. $s$}) guess \stevenc{between} (1) end\stevenc{ing} the side, (2) \stevenc{returning $N(l,s)$} or (3) \stevenc{putting} $l$ somewhere in $U$.
\caption{
addOnSide*($N,s$)\label{algo:addOnSide}
}
\end{algorithm}
\begin{proof}
Note that Algorithm \ref{algo:addOnSide}, i.e. Algorithm addOnSide*, coincides with Algorithm 1 of \cite{softwiredClusterFPT} but for line \stevenc{15} and \stevenc{line 53}. Thus, we only detail the modified lines in the pseudocode.
We stress here that reading \cite{softwiredClusterFPT} is a prerequisite for the comprehension of this \stevenc{and
subsequent lemmas}.

Most of the proof of this lemma coincides with the proof of Lemma 3 of \cite{softwiredClusterFPT}.  
\stevenc{This is true because a network displaying a set of trees $\cT$ always represents the set of clusters $Cl(\cT)$. (Unfortunately the opposite direction is not always true \cite{twotrees}). Hence, whenever the original algorithm rejects a candidate solution because it does (or will not be able to) represent the clusters $Cl(\cT)$, we may
conclude that this candidate solution certainly did not (or will not be able to) display all the trees in $\cT$, and thus also reject it from consideration.}
%which reject a candidate solution because it d \stevenc{and
%much of the original algorithm is bas
%
%(however, the opposite is not true). \stevenc{In particular, when the original algorithm rejects a candidate solution %because
%it does not 
Note that Lemma 3 of \cite{softwiredClusterFPT} has been implicitly extended to also apply to ST-collapsed cluster sets and $r$-reticulation generators in Section 4 of \cite{softwiredClusterFPT}, by extending Propositions 1 and 2, and Observation 5. 

Only two parts of the proof of  Lemma 3 of \cite{softwiredClusterFPT}  do not work for ST-collapsed sets of \emph{trees}. \stevenc{Having modified the original algorithm to obtain Algorithm \ref{algo:addOnSide}, we now detail how to also modify the relevant parts of the original proof}.

\paragraph{\textbf{Case $\boldsymbol{U = \emptyset}.$}} 

The only part of the proof of Lemma 3 of \cite{softwiredClusterFPT} that does not hold here
concerns the situation encountered when the original line 15 is reached:  $|L'| = 1$, $B(l) = \emptyset$, and  ${N(l,s)}$   \emph{does} \stevenc{represent $\cC$} restricted to $\cX(N)  \cup \{l\}$. At this point the new line 15 applies.  

It has been proven in Lemma 3 of \cite{softwiredClusterFPT} that in any valid completion of $N$  there can be no taxon $l'\neq l$ directly above $x_i$ on $s$. So all valid completions terminate the side at $x_i$ or have $l$ above $x_i$.
%So, returning the network $N(l,s)$ can be wrong only if all valid completions of $N$ terminate %the side $s$ at $x_i$. 
%In the following we will show that it is often possible to deterministically predict which
%of the two situations holds. In cases where there is nevertheless ambiguity as to which of the two situations holds, we can address this by returning both solutions (i.e. branching in the
%search tree). Under the assumption that $\cT$ is well-bounded
%this will still allow us to have a running time of $f(r).poly(m)$.
%
We will perform several tests, one at a time, and we only resort to returning
both solutions (i.e. branching) if the final part of line 15 is reached.

The first test is: \emph{if $N(l, s)$ does not display each tree in $\cT$ restricted to $\cX(N)  \cup \{l\}$ then declare $s$ as finished and return $N$} (see Algorithm 1, line 15, first sentence). If $N(l,s)$  does not display each tree in $\cT$ restricted to $\cX(N)  \cup \{l\}$ then it is not possible to construct a valid completion from $N(l,s)$. Hence the only option that remains in this case is to terminate the side.

 The second test \stevenc{(Algorithm 1, line 15, second sentence)} is simple. \emph{If $x_i$ is currently the only taxon on side $s$, then
we only have to return $N(l,s)$}. This is correct because of the assumption that $s$ is a long side (i.e. has at least two taxa)
combined with the earlier observation that the only taxon that can appear above $x_i$ on side
$s$ (if any) is $l$. We may thus henceforth assume that there is at least one taxon underneath $x_i$ on side $s$. Before discussing the third test, we need some observations about the relative location of $x_i$ and $l$ in the input trees.
\begin{figure}[h]
\centering
\includegraphics[scale=.3]{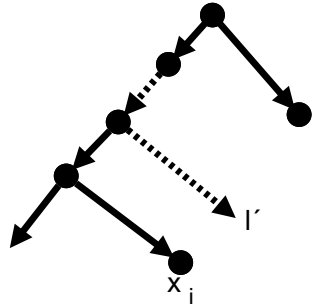}
\caption{If the unique directed path from the parent of $l$ to the parent of $x_i$ contains an interior node then taxon $l'$ must exist. But $N(l,s)$ cannot possibly display such a tree.}
\label{fig:NoInterior}
\end{figure}

Let $T$ be a tree of $\cT$. We may assume that there is a directed path in $T$ (possibly of length 0) from the
parent of $l$ to the parent of $x_i$. Indeed, if this was not the case we would have that either the parent of $x_i$ is an ancestor of  the parent of $l$ or these two parent nodes are not comparable \stevenc{i.e. neither is related to the other by the ancestor-descendant relation}. In both cases $T$ would contain a non-singleton cluster containing $l$ but not $x_i$, but this is impossible since $l \rightarrow_{Cl(\cT)}  x_i$. Moreover,  the directed path from the parent of $l$ to the parent of $x_i$ cannot contain any interior nodes. To see this, suppose there was some interior node on this path; this would imply the existence of a taxon $l'$ lying
``strictly between'' $x_i$ and $l$ (see Figure \ref{fig:NoInterior}). \cs{Moreover, since $B(l) = \emptyset$, we have that $l'$ is in $\cX(N)$.} But then $T$ cannot possibly be displayed by $N(l,s)$ because $N(l,s)$ leaves no space for $l'$ to be in the correct position. This contradicts the assumption
that $N(l,s)$ displays all the trees in $\cT$ restricted to $\cX(N)  \cup \{l\}$.
%and let $e$ be any outgoing edge of $p$ that does \emph{not} lie on the unique path in $T$ from the parent
%of $l$ to the parent of $x_i$. Let $l'$ be any taxon reachable in $T$ by a directed path starting from $p$ and using $e$ as its first edge.
%Clearly, $l' \in S(l)$. 
%we would have that there exists a taxon $l' \in S(l)$.
% and a cluster $C$ of  $Cl(T)$ such that $x_i$ but not $l'$ is in $C$.
%However, since $B(l)=\emptyset$, $l'$ should already be in $\cX(N)$.
%Note that 
% a contradiction.
Combining these insights we see that there are only two possible configurations for $T$, Cases A and B, depicted in Figure \ref{fig:figCasesAB}.

\begin{figure}[h]
\centering
\includegraphics[scale=.3]{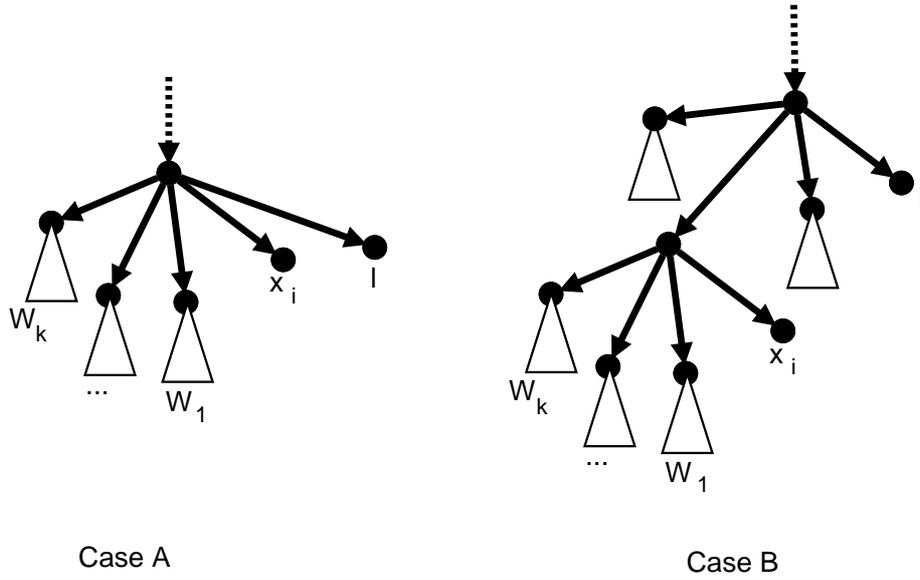}
\caption{Case A is when $l$ and $x$ have the same parent, Case B is when the parent of $l$ is the grandparent of $x_i$.
% These are the only two topologies that an input tree can have
%when $l \rightarrow_{Cl(\cT)}  x_i$ and ${N(l,s)}$ \emph{does} \stevenc{display} each tree in %$\cT$ restricted to $\cX(N)  \cup \{l\}$
}
\label{fig:figCasesAB}
\end{figure}

Irrespective of whether the tree $T$ is in Case A or Case B
we require the following definitions. Let $p$ be the parent of $x_i$. Let $u_1, \ldots, u_k$
$(k \geq 1)$ be the children of $p$ not equal to $x_i$. (Note that $k=1$ if $T$ is binary). For each $u_j$ let $W_j$ be the subtree of $T$ rooted at $u_j$ (i.e. \cs{a} sibling subtree of $x_i$), and let $\cX(W_j)$ be the set of taxa in $W_j$. 
Let $\cW$ be the union of all the $\cX(W_j)$. Observe that in Case B, $\cW \subseteq \cX(N)$ i.e. all taxa in $\cW$ have already been allocated. (If this was not so, $B(l) \neq \emptyset$ and we would not be in this case anyway). Observe also that in Case B all the taxa in $\cW$ are
reachable in $N$ by directed paths from the parent of $x_i$. If this was not so, then
$N(l,s)$ would not have displayed all the trees in $\cT$ restricted to $\cX(N) \cup \{l\}$ (specifically: $T$) and
we would not be in this case anyway.

We say a tree $T \in \cT$ is  \emph{safe} \cs{w.r.t. $s$}
if (i) it is in Case A or (ii) it is in Case B and for each of its $W_j$, at least one taxon from $\cX(W_j)$ has already been allocated to side $s$.
%\marginpar{If we re-use this proof for the case
%$U \neq \emptyset$ then note that in that case a binary tree is \emph{not} always safe, but this doesn't hurt us in anyway.} A little thought should make clear that a binary
%tree $T$ is always safe: if it is not already in Case A, then the taxon below $x_i$ on side $s$ (which definitely exists because we passed the first test) is definitely in $\cX(W_1)$, and
%$W_1$ is the only sibling subtree of $x_i$ in $T$.
If a
tree is not safe \stevenc{w.r.t. $s$} then it is \emph{unsafe} \cs{w.r.t. $s$}: it is in Case B and there exists at least one $W_j$ such that none of the taxa in $\cX(W_j)$ have been allocated to side $s$. (Combining this fact with
the earlier observations that all the taxa in the sibling subtrees of $x_i$ have already been allocated and are reachable by directed paths from the parent of $x_i$, we note that \stevenc{in this unsafe situation} all the taxa in $W_j$ must
have been allocated to sides 
\cs{reachable from} \stevenc{i.e. ``underneath''}
%``underneath''
side $s$).

 The third test \stevenc{(Algorithm 1, line 15, third sentence)} is this:
\emph{if all trees in $\cT$ are safe \cs{w.r.t. $s$}, then return $N(l,s)$}.
%(This test will be automatically
%passed if all trees in $\cT$ are binary).
We now argue that this is correct.
Suppose, for the sake of contradiction, that all valid completions
terminate the side at $x_i$. (As mentioned earlier it is not possible that a taxon $l'$ other than $l$ is placed immediately above $x_i$). Let $N'$ be an arbitrary valid completion of $N$. Note that, by definition,
$N'$ has the same set of side guesses as $N$. Let  $T_{N'}$ be the switching of $N'$ \stevenc{corresponding to} a binary refinement of $T$. Denote by $N''$ and $T_{N''}$ respectively the network and  the switching obtained respectively from $N'$ and   $T_{N'}$ by moving $l$, wherever it is, onto the side $s$,  just above $x_i$. We claim that $N''$ displays all the trees in $\cT$.
It is not too difficult to see that, if $T$ is in Case A, $T_{N''}$ still \stevenc{corresponds to} a binary refinement of $T$. This also holds if $T$ is in Case B and
for each $W_j$ at least one taxon from $\cX(W_j)$ is on side $s$; the central argument for this
is that in any switching  in a valid completion \stevenc{that corresponds to a} binary refinement of $T$, the parent of $x_i$
will always be the lowest common ancestor of $\cW \cup \{x_i\}$. Furthermore, we can argue as in \cite{softwiredClusterFPT} that because of the assumed minimality
of the side guesses, $N''$ has the same set of side guesses as $N$ and $N'$, see \cite{softwiredClusterFPT} for the full argument. Hence we can conclude that $N''$ is a valid completion, yielding a contradiction. So if all trees in $\cT$ are safe \stevenc{w.r.t. $s$}, then returning $N(l,s)$ is correct. 

If we have reached this point then at least one tree in $\cT$ is unsafe \cs{w.r.t. $s$}. The problem we
face here is that for an unsafe tree $T$ it might hold that in all switchings \stevenc{corresponding to a binary refinement} of $T$, ranging over all valid completions, the lowest common ancestor
of $\cW \cup \{x_i\}$ lies \emph{above} side $s$. In such a switching moving $l$ directly
above $x_i$ creates a \stevenc{switching that does not correspond to a binary refinement of $T$}, because $l$ will wrongly have been
put ``inbetween'' the $W_j$. Hence, we cannot be certain that only returning $N(l,s)$
is correct, because this might make it impossible to reach any valid completions.

Hence, we \textbf{guess} \stevenc{(Algorithm 1, line 15, fourth sentence)}. \emph{That is, we return both $N(l,s)$ and $N$ with the side terminated just above
$x_i$}. This is correct (because these are the only two possibilities) but it causes the
search tree to branch. \stevenc{Unfortunately,} in a valid completion
there might be $O(|\cX|)$ taxa on side $s$, and in the worst case we might have to
branch (because some tree is unsafe) for each taxon as it is placed on side $s$. This
could inflate the running time by a factor of $2^{O(|\cX|)}$, which is in \stevenc{general} not
$f(r).poly(m)$\footnote{\stevenc{Note, however, that if $|\cT|$ is exponentially large as a function of $|\cX|$, $2^{O(|\cX|)}$ becomes $poly(m)$, meaning that in such cases $f(r).poly(m)$ running time might still be possible.}}. However, under the assumption of well-boundedness we can
guarantee \cs{a} $f(r).poly(m)$  \cs{running time}, as we will now show. 
The key to proving this lies in proving two observations:\\\\
\cs{\textbf{Observation 1}}. \emph{During the execution of Algorithm 2 \stevenc {(which repeatedly calls Algorithm 1)} a tree $T$ can be unsafe \cs{w.r.t. $s$} \emph{at most once}. Specifically, all trees that are unsafe at
the moment we branch, will never be unsafe again \cs{w.r.t. $s$}.
Hence, each time we branch at least one tree in $\cT$ will become safe \cs{w.r.t. $s$} for the remainder of the execution. If $|\cT|$ is bounded by $f(r)$ (condition (1) of well-boundedness) then the inflation in the running time caused by branching will thus be limited to $2^{f(r)}$, which is itself $f(r)$; after this point Algorithm 1 will never branch again.}
\\\\
\begin{figure}[h]
\centering
\includegraphics[scale=.3]{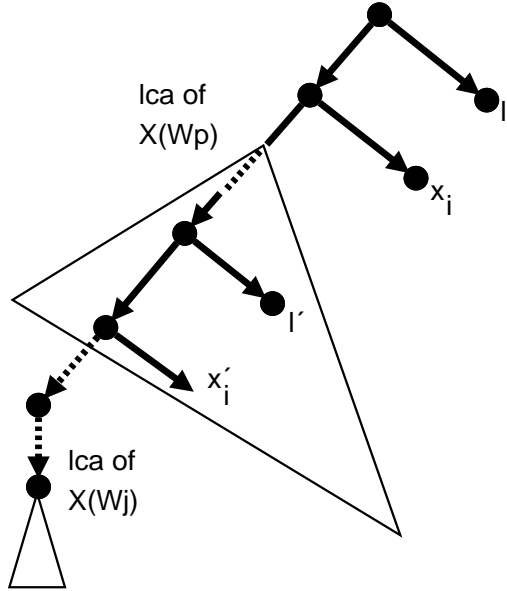}
\caption{If a tree is unsafe  \cs{w.r.t. $s$} a second time, and the directed path from the parent of  $x_i$ (in $T$) to the parent of
$l'$ has length greater than zero,  then $x'_i$ and $l'$ belong to the same $W_p \neq W_j$, such that the lca
of $\cX(W_p)$ (in any \stevenc{binary refinement} of $T$ displayed by the network) is an ancestor of the lca of $\cX(W_j)$, contradiction.}
\label{fig:safe1}
\end{figure}
We now prove this observation. Suppose
that it does not hold, and that some tree $T$ becomes unsafe \cs{w.r.t. $s$} a second time. Let $x'_i$ and
$l'$ be the corresponding taxa the first time $T$ was unsafe  \cs{w.r.t. $s$} and let $x_i$ and $l$ be
the taxa, and $N$ the incomplete network, at the second moment of unsafeness, \cs{see Figure \ref{fig:safe1}}. Clearly, at the
second point of unsafeness $x'_i, l'$ and $x_i$ are already on side $s$, in that
relative order. (We include the possibility that $l' = x_i$).
Note that in $T$ the parent
of $l$ is a strict ancestor of the parent of $x_i$, 
and the parent of $l'$ is a strict ancestor of the parent of $x'_i$. (This follows
because unsafeness implies Case B). Furthermore, the
fact that $N(l,s)$ displays all trees restricted to $\cX(N) \cup \{l\}$ means that in $T$ there
is a directed path (possibly of length 0) from the parent of $x_i$ to the parent of $l'$.
Whichever holds, the parents of $l$ and $x_i$ are strict ancestors of the parent of $x'_i$ in $T$.
Now, given that $T$ is unsafe  \cs{w.r.t. $s$} for a second time, at least one of the $W_j$
corresponding to $x_i$ (i.e. the sibling subtrees of $x_i$ in $T$) is such that all the
taxa in $\cX(W_j)$ lie strictly underneath side $s$. Also, we know that $N(l,s)$ displays
$T$ restricted to $\cX(N) \cup \{l\}$. Consider any \stevenc{binary refinement} of $T$ restricted
to $\cX(N) \cup \{l\}$ displayed
by $N(l,s)$, and observe that in such a refinement the parent of $x'_i$ lies on
the directed path from the parent of $x_i$ to the lowest common ancestor of $\cX(W_j)$.
(This directed path must exist because of the location of $l$ just above $x_i$ in $N(l,s)$).
Clearly, neither $x_i$ nor $l$ is in $\cX(W_j)$.
%\emph{(The following is hyper-technical...have I missed any cases...?)}
Recall that (by the definition of Case B) in $T$ the grandparent of $x'_i$ is the same node as the parent of $l'$. So, if the directed path from the parent of $x_i$ (in $T$) to the parent of
$l'$ has length greater than zero, then $x'_i$ and $l'$ belong to the same $W_p \neq W_j$
of $x_i$. See Figure \ref{fig:safe1}.

 But this cannot be so, because in the \stevenc{binary refinement} the lowest common ancestor
of $W_j$ would be an ancestor of the lowest common ancestor of $W_p$, which is not
allowed. (It is not allowed because, whichever binary refinement we choose, all the sibling subtrees of $x_i$ should be incomparable). So suppose  that the parent of $l'$ is the same as the parent of $x_i$ in $T$. In this case  $l'$ and $x'_i$ cannot belong to the same \stevenc{sibling subtree} of $x_i$, because (by Case B)
the parent of $x'_i$ in $T$ is not the parent of $l'$.
Suppose $x'_i$ is in $\cs{W_q \neq W_j}$. Then $|\cX(W_q)| \geq 2$, \cs{otherwise it would not be possible for the pair $x'_i,l'$ to cause the unsafeness of  side $s$ in the previous iteration}. But then we have that, in the \stevenc{refinement}, the lowest common
ancestor of $W_q$ lies on the edge-side \stevenc{$s$} itself i.e. is an ancestor of the lowest common ancestor of $W_j$ (see Figure \ref{fig:safe2}), which is
again not allowed. (The reason is the same as before: in any binary refinement the lowest common
ancestors of all the \stevenc{sibling subtrees of $x_i$} should be mutually incomparable). From this we conclude that $W_j$ cannot lie entirely underneath
side $s$. Hence some taxon of $W_j$ is already on side $s$. Hence $T$ is safe \cs{w.r.t. $s$}. Hence $T$ could not become unsafe  \cs{w.r.t. $s$} for a second time.

\begin{figure}[h]
\centering
\includegraphics[scale=.3]{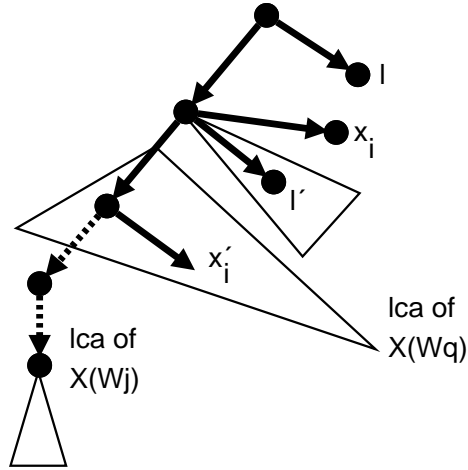}
\caption{If a tree is unsafe  \cs{w.r.t. $s$} a second time, and the directed path from the parent of  $x_i$ (in $T$) to the parent of
$l'$ has zero length,  then $x'_i$ is in some $W_q \neq W_j$ with at least one other taxon, meaning that the lca
of $\cX(W_q)$ (in any \stevenc{binary refinement} of $T$ displayed by the network) is an ancestor of the lca of $\cX(W_j)$, contradiction.}
\label{fig:safe2}
\end{figure}

So, we have shown that an $f(r)$ bound on
$|\cT|$ is indeed sufficient to get the running time we need. \cs{We now introduce the second key observation that concerns the
case when node degrees are bounded by $f(r)$: }\\\\
\cs{\textbf{Observation 2}}. \emph{Let $d$ be the maximum node degree ranging over
all nodes in all trees in $\cT$. \cs{Then}, after \stevenc{at most} $d+1$ \cs{branchings}, all trees in $\cT$ will have become safe  \cs{w.r.t. $s$}}\footnote{\stevenc{We note that the definition of $d$, and the choice of $d+1$,  is not particularly well-optimized. In
particular, in the case $U = \emptyset$ branching can only happen if there are already 2 or more taxa allocated to side $s$; the placement of the first two taxa is deterministic (thanks to the hard assumption that $s$ is a long side). An interesting consequence of this is that, if $\cT$ is binary, the $U = \emptyset$ phase of the algorithm is \emph{entirely} deterministic, because after placing the initial two taxa on the side no tree in $\cT$ can subsequently ever be unsafe w.r.t. $s$.}}.\\

 Suppose then that we have already branched $d+1$ times. This means that there are already
$d+1$ taxa at the bottom of side $s$. Now, suppose a tree $T$ is unsafe  \cs{w.r.t. $s$}, so some
$W_j$ lies entirely underneath side $s$. We know that $N(l,s)$ displays $T$ restricted
to $X(N) \cup \{l\}$, from which we can conclude that all the \cs{taxa} on side $s$ below
$x_i$ belong to (possibly different) \stevenc{sibling subtrees} of $x_i$ in $T$. Observe that it cannot happen that two or more
of the taxa below $x_i$ on side $s$ belong to the same \stevenc{sibling subtree $W_p$}. This holds because
it would mean that in any \stevenc{binary refinement} (\stevenc{displayed by} a valid completion of $N(l,s))$
the lowest common ancestor of $W_p$ is an ancestor of $W_j$, and this is not allowed
because they should be incomparable. So the only way $T$ can be unsafe is if every
taxon on $s$ below $x_i$ is in a different \stevenc{sibling subtree}. But, because of the degree bound,
there can only be at most $d$ different \stevenc{sibling subtrees}, so this is not possible. Hence $T$ is safe,
contradiction.

This concludes the proof that the $U = \emptyset$ case terminates after at most
$f(r).poly(m)$ iterations, assuming well-boundedness.

%Now, recall that 
%Also, $x_i$ and $l$ must both belong
%to the taxa of some $W_i \neq W_j$. (

\paragraph{\textbf{Case $\boldsymbol{U \neq \emptyset}.$}}  
\stevenc{With the exception of} the subcase encountered when the original line 53 is reached - i.e. when 
simultaneously $|L'|=1$, $B(l)=\emptyset$, $N(l,s)$  \emph{does} represent ${\cC}$ restricted to $\cX(N) \cup \{l\}$ and $N^{*}(l,s)$  does represent $\cC^{*}$ - all cases can be proven as argued in Lemma 3 of \cite{softwiredClusterFPT}. \stevenc{In this subcase
the new line 53 of Algorithm 1 applies, and we now prove its correctness.}

In this remaining subcase there are (at most) three possibilities. (1) The
side terminates\footnote{\cs{Recall that}, if $x_i$ is the only taxon currently allocated to side $s$ then this possibility is
excluded, because it violates the assumption that $s$ is a long side.} at $x_i$;\ (2) $l$ is the taxon immediately above $x_i$ on side $s$; (3) $l$
is on some side in $U$. Observe that this really covers all cases. If the side does not
terminate at $x_i$, $l$ is not in $U$ and $l$ is not immediately above $x_i$, then some
not yet allocated  taxon $p \neq l$ must be immediately above $x_i$. But then we would have $p \stevenc{\rightarrow_{Cl(\cT)}} x_i$ (so $p \in L$) and $l \stevenc{\rightarrow_{Cl(\cT)}} p$, so $l \not \in L'$, contradiction. 

We now prove that the sequence of steps shown in the new line 53 is correct.
Firstly, suppose $N^{*}(l, s)$ does \emph{not} display each tree in $\cT$ restricted to $\cX(N)  \cup \{l\}$. Then (2) is excluded as a possibility. So in this case we guess (1) or (3) i.e. guess either to end the side or to put $l$ somewhere in $U$. Note that each time (1) or (3) is guessed an $f(r)$-counter is decremented. This is because
the number of sides is $f(r)$-bounded \cs{and  either $s$ is declared as finished or a short side of $U$ is filled}.

We may henceforth assume that
$N^{*}(l, s)$ \emph{does} display each tree in $\cT$ restricted to $\cX(N)  \cup \{l\}$. We
observe that if we have reached this point then every tree in $\cT$ will be in Case A or Case B; the proof of this
given in the case $U = \emptyset$ goes through here too. (The two comments about trees in Case B that follow
this proof also still hold). The notion of \emph{safe} and \emph{unsafe} is hence still well-defined. In fact, the
proof that - when all trees are safe - it is legitimate to simply return $N(l,s)$ also holds. So the only situation left to
consider is when at least one tree is unsafe. As argued
above there are only three possibilities for action and in line 53 we consider all of them. From this we conclude that
the algorithm is correct. However, it is still necessary to bound the running time.

The only ``dangerous'' guess is (2) because unlike (1) and (\cs{3}) this does not obviously decrement any $f(r)$-bounded counters. To show that this 
\cs{still leads us to a  $f(r) \cdot poly(m)$ running time} %does not damage us too much 
we will prove that it can happen at most once that
a tree is unsafe \emph{and} we subsequently guess (2). The proof is unchanged from $U = \emptyset$. The only difference, and the reason that we emphasize the \emph{and}, is that a tree might
be unsafe but (rather than putting $l$ above $x_i$) we decide to put $l$ in $U$, meaning that the same tree
can still be unsafe again in a later iteration. However, due to the $f(r)$-bound on the number of sides this cannot happen too often (and the case $U = \emptyset$ will be reached).
Combining these insights shows that if there are at most
$f(r)$ trees in $\cT$ then we will reach the last part of line 53 at most $f(r)$ times in total \stevenc{(during the construction
of side $s$)}. Alternatively, if the maximum degree of trees in $\cT$ is bounded by $f(r)$, then (just as in the
case $U = \emptyset$) we can argue that if $d+1$ taxa have already been placed on side $s$, and we have survived the first check in line 53, then \emph{all} the trees in $\cT$ will be safe  \cs{w.r.t. $s$} and
it is fine to only return $N(l,s)$.\\
\\
This concludes the proof of the lemma. 
\end{proof}

\begin{algorithm}[]
$ \mathcal{N} \gets{N}$\;
\While{there exists $N \in \mathcal{N}$ such that $s$ is not finished in $N$}{

		$\mathcal{N} \gets{} \mathcal{N}  \setminus N$\;
		$\mathcal{N} \gets{}$ $\mathcal{N}$ $\cup$ addOnSide*$(N, s)$\;

}
\caption{
%Subroutine for choosing taxon $x_{i+1}$ on side $s$ or deciding to \textsc{END SIDE}
completeSide*($N,s$)\label{algo:completeSide}
}
\end{algorithm}

\begin{lemma}
\label{lemma:subroutine2IsCorrect}
\cs{Let $\cT$ be a \cs{well-bounded} ST-collapsed set of trees on $\cX$ and let $r$ be the first integer such that a network with reticulation number $r$ displaying all the trees in $\cT$ exists.
Let $N$ be an incomplete network such that 
its underlying $r$-reticulation generator $G$ and set of side guesses $S_G$ are such that $(G,S_G)$ is side-minimal w.r.t. $\cT$ and $r$, and let $s$ be an active side of $N$. 
Then, if a valid completion for $N$ exists, Algorithm \ref{algo:completeSide}  computes a set of (incomplete) networks $\mathcal{N}$ such that  
 this set contains at least one  network  for which a valid completion exists {\bf for which $s$ is a finished side}  \cs{in time $f(r) \cdot poly(m)$,  where $m=|\cX|+|\cT|$}. 
 }
\end{lemma} 
\begin{proof}
\cs{Algorithm \ref{algo:completeSide} is the same as Algorithm 2 of \cite{softwiredClusterFPT}, but for the fact 
that in the former the subroutine addOnSide* (defined in Algorithm \ref{algo:addOnSide}) is 
called, rather than the subroutine addOnSide (defined in Algorithm 1 of \cite{softwiredClusterFPT}). 
From this observation and by Lemma \ref{lemma:subroutineIsCorrect}, which extends Lemma 3 of \cite{softwiredClusterFPT}, the proof of Lemma 4 of  \cite{softwiredClusterFPT}  
can be easily adapted to prove this lemma. }
\end{proof}

\begin{algorithm}[] 
\nlset{1} {\bf foreach}  $r$-reticulation generator G in increasing side order {\bf do}\\
\nlset{7} $\mathcal{N'}\gets{}$ completeSide*($N(s^{-},s),s$)\\
\nlset{16} {\bf if} there is a network $N' \in \mathcal{N'}$ displaying all the trees in $\cT$ \textbf{then}
\caption{
ComputeNetwork*($\cT$)\label{algo:main}
}
\end{algorithm}

\begin{lemma}
\label{lem:core}
Let $\cT$ be a well-bounded, ST-collapsed set of trees on $\cX$. Then, for every fixed
$r \geq 0$, Algorithm \ref{algo:main} determines  whether a  network $N$ such that 
$r(N)=r$ displaying every tree in $\cT$ exists, and if so, returns it  in time $f(r) \cdot poly(m)$,  where $m=|\cX|+|\cT|$. 
\end{lemma}
\begin{proof}
\cs{Algorithm \ref{algo:main} coincides with Algorithm 3 of \cite{softwiredClusterFPT} but for lines 1, 7 and 16. Thus, we only detail the modified lines in the pseudocode.
%Recall that a network displaying a set of trees $\cT$ always represents the set of clusters $Cl(\cT)$. 
%Note that Lemma \ref{lemma:subroutineIsCorrect} extends Lemma 3 of \cite{softwiredClusterFPT}. 
%Moreover, by Lemma  \ref{lemma:subroutineIsCorrect}, Lemma 4 of \cite{softwiredClusterFPT} is valid also for tree displaying of ST-collapsed  set of  trees.  
Since by Lemma \ref{lem:genIsSufficient_bis}, we can narrow the search to the  set  $\mathcal{\hat{N}}^{r}_{\cX}$, by Lemma \ref{lemma:subroutine2IsCorrect} we can use the same \stevenc{proof scheme} as Lemma 5 of \cite{softwiredClusterFPT} to prove the lemma  - the check of line 16 can be done in time $f(r) \cdot poly(m)$ by Proposition \ref{prop:checkReprTree}.%can
%determine whether a  network $N$ such that 
%$r(N)=r$ representing $Cl(\cT)$ exists, and if so, returns it  in time $f(\cs{r}) \cdot poly(m)$. 
%So, if we add a check after line 16 of Algorithm 3 of \cite{softwiredClusterFPT} establishing whether $N$ also displays each tree in $\cT$ or not, we have an algorithm solving the problem stated in this lemma.  
}
\end{proof}

From Lemmas \ref{lem:stcollapse1}, \ref{lem:hybridNumInvariant}   and  \ref{lem:core}, we \stevenc{finally} conclude the following:

\begin{theorem}
\label{thm:retics}
Let $\cT$ be a  well-bounded set of trees on $\cX$. Then, for every fixed
$r \geq 0$, it is possible to determine in time $f(r)  \cdot poly(m)$, where $m=|\cX|+|\cT|$,  whether a network that displays all the trees in $\cT$ with reticulation number at most
$r$ exists (and if so, to return such a network).
\end{theorem}

\section{Minimizing the reticulation number of a set of trees is polynomial-time solvable
for a fixed number of reticulations}

For completeness we show that, even though we do not yet have an FPT result for an
arbitrary set of nonbinary trees on the same set of taxa $\cX$, we do have the following weaker result which shows that
for a fixed number of reticulations the problem is polynomial-time solveable. This is strictly more general
than the (implied) polynomial-time results in \cite{ierselLinz2012} and \cite{linzsemple2009} due to the fact
that here it is permitted to have an unbounded number of non-binary trees in the input.

\begin{theorem}
\label{thm:polyNotFPT}
Let $\cT$ be a set of trees on $\cX$. Then, for every fixed
$r \geq 0$, it is possible to determine in polynomial time - specifically, time $O( m^{f(r)} )$,
where $m=|\cX|+|\cT|$ - whether a network that displays all the trees in $\cT$ with reticulation number at most
$r$ exists (and if so, to return such a network).
\end{theorem}
\begin{proof}
The proof is a straightforward extension of Theorem 1 from \cite{elusiveness} adapted to use $r$-reticulation generators. We sketch the construction here.
Without loss of generality we can assume that $\cT$ is ST-collapsed and that we can restrict our attention to adding taxa to
$r$-reticulation generators. We can also assume without loss of generality that
$r(\cT) \geq r$. Let $N$ be a network that displays all the trees in $\cT$, such that $r(N)=r(\cT)$.
We begin by guessing the correct $r$-reticulation generator for $N$; there are at most $f(r)$ such
generators, and an $r$-reticulation generator has at most $O(r)$ sides. For each side of the generator we guess whether it has 0,1,2 or $> 2$  taxa on it.
For sides with 1 taxon we guess the identity of that taxon. For sides $s$ with $\geq 2$ taxa
we guess the identity of the taxon nearest the root on that side, $s^{+}$, and the taxon furthest
from the root on that side, $s^{-}$. We say that a $> 2$ side
is \emph{lowest} if it does not yet have all its taxa and there
is no other $>2$  side $s' \neq s$ with this property that is reachable by a directed path from the head of $s$.
The algorithm chooses a lowest side and adds all its taxa to it, repeating this until there are no more lowest
sides (i.e. until all $|\cX|$ taxa have been added to the network). When the sides are processed in this order,
a taxon $x$ belongs on side $s$ if and only if there is a cluster $C \in Cl(\stevenc{\cT})$ such that $x$ and $s^{-}$
are both in $C$, but $s^{+}$ is not. Now, once all the taxa for a $> 2$ side $s$ have been identified,
their order on that side is uniquely identified. This follows because the $\rightarrow_{Cl(\cT)}$ relation
imposes a total order $s^{+} \rightarrow_{Cl(\cT)} \ldots \rightarrow_{Cl(\cT)} s^{-}$ on the
taxa. (Note that there cannot be cycles in the $\rightarrow_{Cl(\cT)}$ relation because the
taxa in the cycle would then induce an ST-set, contradicting the assumption that $\cT$ is ST-collapsed  \cite[Proposition 2 (extended)]{softwiredClusterFPT}).
Finally, having added all the taxa to the $r$-reticulation generator we can check in time $f(r).poly(m)$ whether
it displays all the trees in $\cT$. In this way we can identify $N$ in polynomial-time.
\end{proof}

\section{Conclusions and future directions}

\stevenc{
In this article we have described quite broad sufficient conditions under which the computation of reticulation number of a set $\cT$ of trees is FPT. This extends
existing FPT results, which applied to two specific cases: (i) two nonbinary trees, and (ii) an arbitrarily large set of binary trees. The obvious open question that remains is whether computation of reticulation number is still FPT when the ``well-bounded'' condition is lifted i.e. when there are an unbounded number of nonbinary trees in the input
with unbounded maximum degree. We note already the following \stevenc{three} observations, which seem to be important in this regard.
\emph{Primo}, if $|\cT| > 2^{r}$ and $r(\cT) \leq r$ then two or more trees in $\cT$ must have a common binary refinement (because a network with $r$ reticulations can display at most $2^{r}$ distinct trees). 
\emph{Secundo},
a $r$-reticulation generator has at most $5r$ sides i.e. $r$ node sides and at most $4r$ edge sides \cite{softwiredClusterFPT}. Hence, if a tree $T \in \cT$ has a node $u$ such that $u$ has
more than $5r$ children,  then at least two of the taxa reachable by directed paths in $T$ from $u$, \emph{must}
be on the same edge-side of the underlying $r$-reticulation generator in any valid completion.  In \cite{ierselLinz2012} (i.e. when $\cT$ is binary) a similar argument is used to develop a kernelization strategy in which \emph{common chains} have length at most $f(r)$. Specifically, chains with length longer than $5r$ must have at least two taxa allocated to the same edge side, from which can immediately be concluded that the entire common chain can be safely allocated to that side. Unfortunately, in the case of multiple nonbinary trees it is still not entirely clear how \emph{common chains} should be defined and utilized. In particular, it is not clear how to generalise the definition given in \cite{linzsemple2009} for two trees, to the case of multiple trees, such that an FPT running time is obtained \emph{even when the number of trees in the input
is unbounded}. 
\emph{Tertio}, there has so far been little attention for the possibility that the problem is not FPT. The standard
way of proving non-FPT is to show that a problem is $W[1]$-hard \cite{niedermeier2006,Flum2006}. It would be interesting to explore this
possibility further, which would require developing FPT-reductions from (for example) independent set or maximum
clique. Such problems have not yet figured prominently in the phylogenetic network literature, which makes this an
interesting research direction in its own right.
}

\stevenc{
\section{Acknowledgements}
We thank Nela Lekic, Simone Linz and Leo van Iersel for useful discussions.
}
% only $f(r)$ of the subtrees rooted at the children of a node $u$ can be non-taxa; the
%rest must be singleton subtrees (i.e. taxa). Again, this is because of the $5r$ upper bound on the number of sides in the %generator. 

%\huge{This is the most complex article ever}\\
%\huge{It's going to be fun finding a journal to publish this}
%\huge{My head is going to explode}

% \newpage


\bibliographystyle{plain}      % basic style, author-year citations
\bibliography{wellbounded}

\newpage
\section{Appendix}

Here we discuss some technical points about (bounds on) the size of the input. The first is a loose upper bound on the number of trees in the input.\\
\\
\noindent
\stevenc{\textbf{Observation 3}}. \emph{Let $N$ be a binary network on $\cX$ with $r$ reticulations. If $|\cT| > 2^{r} \cdot 2^{2(|\cX|-1)}$ then $N$ cannot display all the trees in $\cT$.}
\begin{proof}
$N$ can display up to $2^{r}$ binary trees on $\cX$. By contracting all possible subsets of the edges of a binary tree $T$, we generate all possible
nonbinary trees on $\cX$ of which $T$ is a binary refinement (and some non-valid trees too). There are $2(|\cX|-1)$ edges in a binary tree, from which the claim follows.
\end{proof}

The following toy construction shows how $|\cT|$ can become very large as a function of both $r$ and $|\cX|=n$, without introducing any obvious
redundancy in the input. In particular it shows that even if we assume that $\cT$ is ST-collapsed and that no tree in $\cT$ is a refinement of another, there exist $\cT$ for which
%$r(\cT) = r$ and $|Cl(\cT)| \leq 2^{r+1}(n-1)$ but
$|\cT|$ grows exponentially quickly in both $r$ and $n$.\\
%(The restriction on $|Cl(\cT)| \leq 2^{r+1}(n-1)$ is relevant because this is the maximum number of clusters that a network with $r$ reticulations can represent; violating this bound means that $N$ cannot %possibly represent all the clusters in $Cl(\cT)$, let alone the trees in $\cT$ themselves).\\
\\
\begin{figure}[h]
\centering
\includegraphics[scale=.3]{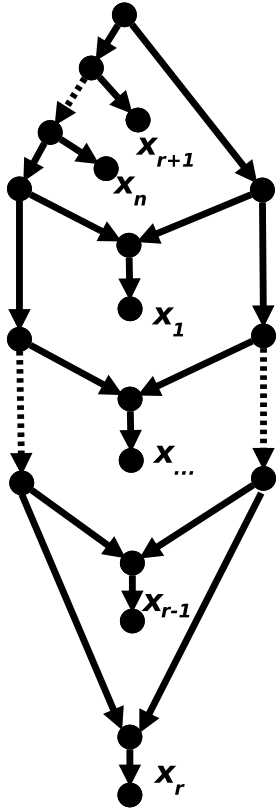}
\caption{A binary network $N$ with $r$ reticulations and $\cX = \{x_1, \ldots, x_{r}, x_{r+1}, \ldots \stevenc{x_n}\}$, where the first $r$ taxa are on node-sides and the remaining $n-r$ taxa
are all on the same edge side.}
\label{fig:exponential}
\end{figure}
Choose $r$ and then choose $n$ sufficiently large with respect to $r$ (we will explain later how to do this). Without loss of generality we assume that $n-r-1$ is odd. Consider the network $N$ shown in Figure \ref{fig:exponential}. Now, construct a set of trees $\cT$ on $\cX$ as follows. Let $B$ be the set of $2^{r}$ binary trees displayed by $N$. For each tree $T' \in B$ we will add a set of nonbinary trees $NB(T')$ to $\cT$, although we will not add
$T'$ itself. Consider the $n-r$ taxa on the \stevenc{top-left} edge side of $N$, and the $n-r-1$ edges that have as endpoints two parents of these taxa. We call these \emph{chain edges}.  The set $NB(T')$ consists
of all trees obtained from $T'$ by contracting exactly $\lceil \frac{n-r-1}{2} \rceil$ of these edges. Hence, 
\[
|\cT| = 2^{r} \binom{\stevenc{n-r-1}}{\lceil \frac{n-r-1}{2} \rceil},
\] which is exponential in both $r$ and $n$.
We first show that there are no trees $T_1, T_2 \in \cT$ such that one is a refinement of the other. To see this, note that if $T_1$ and $T_2$ are in different $NB(.)$ sets, then $T_1$ and $T_2$
cannot be refinements of each other because of the distribution of the taxa $x_1, \ldots, x_r$ in the trees. Suppose then that $T_1$ and $T_2$ both come from the same $NB(.)$ set. Because
of the way we have contracted edges, $T_1$ will have some chain edge that has been contracted in $T_2$, and vice-versa, proving that neither $Cl(T_1) \subseteq Cl(T_2)$ nor $Cl(T_2) \subseteq Cl(T_1)$,
from which we conclude that neither tree is a refinement of the other. Secondly, it can be verified that $\cT$ is ST-collapsed, but we omit the proof.

It remains only to show that $r(\cT)=r$. To argue
this we note that any network with $r$ reticulations can represent at most $2^{r}(2(n-1)-n) + n = 2^{r}(n-2) + n$ clusters. This is the usual cluster bound \stevenc{(see the discussion after Proposition \ref{prop:checkReprTree})} but slightly tightened so that singleton clusters are only counted once. Hence, if we can show that $|Cl(\cT)| >  2^{r-1}(n-2) + n$ then we can immediately conclude that $r(\cT)=r$. (Note that $r(\cT) \leq r$ holds because every tree in $T \in \cT$
is obtained by contracting edges of one of the binary trees $T'$ displayed by $N$ i.e. $T'$ is a refinement of $T$). We claim that $|Cl(\cT)| \geq 2^{r}(n-r-1)$. This holds because, for every
tree $T' \in B$ and every chain edge $e$, there exists some tree in $T \in \cT$ such that $T \in NB(T')$ and the edge $e$ is not contracted in $T$. To conclude, we only have to choose $n$ sufficiently
large that
\[
2^{r}(n-r-1) >  2^{r-1}(n-2) + n.
\]

\noindent
\stevenc{\textbf{Lemma \ref{lem:moveST}.}}
\emph{Given a set of trees $\cT$ on $\cX$, let $N$ be a network displaying all the trees in $\cT$. Let $S$ be a \stevenc{non-singleton} ST-set
with respect to $Cl(\cT)$. Then there exists a network $N'$  displaying all the trees in $\cT$ such that $r(N') \leq r(N)$, 
%$\ell(N') \leq \ell(N)$,
$S$ is under a cut-edge in $N'$ and for each ST-set $S'$ such that $S' \cap S = \emptyset$ and $S'$ is
under a cut-edge in $N$, $S'$ is also under a cut-edge in $N'$.}
\begin{proof}
This lemma follows from the proof of Lemma 11 of  \cite{elusiveness}, where we described how to transform $N$ into a network $N'$ such that $r(N') \leq r(N)$ %, $\ell(N') \leq \ell(N)$,
$S$ is under a cut-edge in $N'$ and for each ST-set $S'$ such that $S' \cap S = \emptyset$ and $S'$ is
under a cut-edge in $N$, $S'$ is also under a cut-edge in $N'$. 
For the sake of completeness, we report here the transformation. 
We obtain $N'$ from $N$ by the following transformation. Let $x$ be any element of $S$ and let $v$ be the node of $N$ labeled by $x$. (a) Delete in $N$ all taxa in $S$
(but not the leaves they label, we will deal with this in step (c)).
(b) Identify $v$ with the root of an arbitrary binary tree $T_{S}$ on $S$ that represents $Cl(\cT)|S$. (c) Tidy up redundant parts of the
network possibly created in step (a) by applying in an arbitrary order
any of the following steps until no more can be applied: deleting any nodes with outdegree-0 that are not labelled
by a taxon; suppressing any nodes with indegree-1 and outdegree-1; replacing any multi-edges with a single edge; deleting
any node with indegree-0 and outdegree-1.

To prove Lemma \ref{lem:moveST}, we still have to show that $N'$ still displays each tree in $\cT$. Let $T$ be a tree of $\cT$ and let $T_N$ be the switching of $N$ \stevenc{corresponding to some refinement of} $T$. In the proof of Lemma 11 of  \cite{elusiveness}, we proved that $T_N$ is modified by the transformation to become a switching $T_{N'}$ of $N'$ such that $T_{N'}$ represents\footnote{Although a switching is not, formally speaking, a phylogenetic tree - because it can have nodes with indegree and outdegree both equal to one, and possibly a redundant node with indegree 0 and outdegree 1 - the definitions of ``represents'' and ``displays'' still hold and behave as expected.} all clusters $C$ of  $T_N$ such that $S \cap C = \emptyset$ or $S \subset C$.
%\marginpar{I think this paragraph needs a bit more work}
Since  $S$ is compatible with $\cC$, the only other clusters of 
$T_N$ are clusters $C$ such that  $C \subseteq S$. Note that by construction $T_{N'}$ displays $T_S$. This implies that $T_{N'}$ displays the set of clusters $Cl(T_N) \setminus (Cl(T_N)|S)  \cup Cl(T_S)$. Note that by construction,  $Cl(T) \setminus (Cl(T_N)|S)  \cup Cl(T)|S=Cl(T)$. Then, since $Cl(T_S) \supseteq Cl(T)|S$ and $Cl(T_N) \supseteq Cl(T)$, it follows that  $Cl(T_{N'}) \supseteq Cl(T) $. So $T_{N'}$ is a \stevenc{corresponds to a refinement} of $T$ and this concludes the proof.
\end{proof}

\end{document}